\documentclass[11pt]{article}
\usepackage{amsthm,amsmath}
\usepackage[bookmarks=true,hypertexnames=false,pagebackref]{hyperref}
\hypersetup{colorlinks=true, citecolor=blue, linkcolor=red, urlcolor=blue}
\usepackage{newpxtext}
\usepackage[libertine,vvarbb]{newtxmath}

\usepackage[T1]{fontenc} 
\usepackage{textcomp} 
\usepackage[scr=rsfso]{mathalfa}
\usepackage{nicefrac}
\usepackage{cleveref}
\usepackage{graphicx}
\usepackage[ruled,vlined,linesnumbered]{algorithm2e}
\usepackage{aliascnt}
\usepackage[section]{placeins}
\usepackage{tcolorbox}
\usepackage{epigraph}
\usepackage{algpseudocode}
\usepackage{xifthen}
\usepackage{latexsym}
\usepackage{framed}
\usepackage{fullpage}
\usepackage{bm}
\usepackage{tikz}
\usepackage[normalem]{ulem}

\newtheorem{theorem}{Theorem}[section]
\newtheorem{claim}[theorem]{Claim}
\newtheorem{fact}[theorem]{Fact}

\newtheorem{corollary}[theorem]{Corollary}
\newtheorem{lemma}[theorem]{Lemma}

\theoremstyle{definition}

\newtheorem{definition}[theorem]{Definition}

\newtheorem{remark}[theorem]{Remark}

\newcommand{\chase}{\textsc{PC}}
\newcommand{\pt}{\textsf{pt}}
\newcommand{\gammavalue}{1 - \frac{0.1}{\log n}}

\newcommand{\size}[1]{\left| #1 \right|}

\newcommand{\set}[1]{\left\{#1\right\}}

\renewcommand{\mid}{\;\middle\vert\;}

\newcommand{\ol}{\overline}
\newcommand{\wt}{\widetilde}

\newcommand{\eqdef}{\stackrel{\mathsf{def}}{=}}

\newcommand{\rv}[1]{\boldsymbol{#1}}
\newcommand{\zo}{\set{0,1}}

\newcommand{\ee}{\mathcal{E}}

\newcommand{\fFam}{\mathcal{F}}

\newcommand{\supp}{\mathrm{supp}}

\newcommand{\msf}{\mathsf}

\newcommand{\true}{\textsc{True}}
\newcommand{\false}{\textsc{False}}

\newcommand{\one}{\mathbb{1}}

\newcommand{\defi}{\mathbf{D}_\infty}

\newcommand{\ent}{\mathbf{H}}
\renewcommand{\Pr}{\operatorname*{\mathbf{Pr}}}
\DeclareMathOperator*{\E}{\mathbf{E}}

\newcommand{\pr}[2][]{ \ifthenelse{\isempty{#1}}
  {\Pr\left[#2\right]} {\Pr_{#1}\left[#2\right]} }
\newcommand{\ex}[2][]{ \ifthenelse{\isempty{#1}}
  {\E\left[#2\right]}
  {\E_{#1}\left[#2\right]} }

\newcommand{\CC}{\mathrm{CC}}

\newcommand{\ie} {i.e.,\ }

\newcommand{\delete}[1]{}
\usepackage[normalem]{ulem}
\usepackage{thm-restate}
\usepackage{enumerate}

\title{Gadgetless Lifting Beats Round Elimination: Improved Lower Bounds for Pointer Chasing} 

\begin{document}

\author{Xinyu Mao \thanks{Research supported by NSF CAREER award 2141536.\\ Thomas Lord Department of Computer Science, University of Southern California.  \\ Email: \texttt{\{xinyumao, guangxuy, jiapengz\}@usc.edu}}
\\
\and
Guangxu Yang \footnotemark[1]
\\
\and
Jiapeng Zhang \footnotemark[1]
}

\maketitle

\begin{abstract}

We prove an $\Omega(n / k + k)$ communication lower bound on \textit{$(k - 1)$-round distributional complexity} of \textit{the $k$-step pointer chasing problem} under \textit{uniform input distribution}, improving the $\Omega(n/k - k\log n)$ lower bound due to Yehudayoff (Combinatorics Probability and Computing, 2020).
Our lower bound almost matches the upper bound of $\wt{O}(n/k + k)$ communication by Nisan and Wigderson (STOC 91). 

As part of our approach, we put forth \textit{gadgetless lifting}, a new framework that lifts lower bounds for \textit{a family of restricted protocols} into lower bounds for \textit{general protocols}.  
A key step in gadgetless lifting is choosing the appropriate definition of restricted protocols. 
In this paper, our definition of restricted protocols is inspired by the structure-vs-pseudorandomness decomposition by G{\"o}{\"o}s, Pitassi, and Watson (FOCS 17) and Yang and Zhang (STOC 24).

Previously, round-communication trade-offs were mainly obtained by round elimination and information complexity. 
Both methods have some barriers in some situations, and we believe gadgetless lifting could potentially address these barriers.
\end{abstract}

\section{Introduction}
\label{sec:intro}

Pointer chasing is a well-known problem \cite{rao2020communication} that demonstrates the power of interaction in communication and has broad applications in different areas.
It was used for proving monotone constant-depth hierarchy theorem \cite{nisan1991rounds,klawe1984monotone}, lower bounds on the time complexity of distributed computation \cite{nanongkai2011tight}, lower bounds on the space complexity of streaming algorithms \cite{feigenbaum2009graph,guruswami2016superlinear,assadi2019polynomial}, adaptivity hierarchy theorem for property testing \cite{canonne2018adaptivity}, exponential separations in local differential privacy \cite{joseph2020exponential}, memory bounds for continual learning \cite{chen2022memory} and limitations of the transformer architecture \cite{peng2024limitations}. 
It is a two-party function defined below.

\begin{definition}[$k$-step pointer chasing function]
For $k\geq 1$, the $k$-step pointer chasing function $\chase_k: [n]^n \times [n]^n \to \zo$ is defined as follows. 
Given input $f_A,f_B \in [n]^n$, for $r = 0, 1, \dots, k$ we recursively define pointers via
$$
    \pt_r(f_A, f_B) \eqdef \begin{cases}
        1 &\text{if $r = 0$}; \\
        f_A(\pt_{r - 1}(f_A, f_B)) &\text{if $r > 0$ is odd}; \\
        f_B(\pt_{r - 1}(f_A, f_B)) &\text{if $r > 0$ is even}.
    \end{cases}
$$
The output of $\chase_k$ is the parity of the last pointer, namely, $\chase_k(f_A, f_B) \eqdef \pt_k(f_A, f_B) \bmod 2$.
\end{definition}


\paragraph{Upper bounds.}
If Alice and Bob could communicate for $k$ rounds, a simple protocol is the following: Alice and Bob alternatively send $f_A(\pt_{r - 1}(f_A, f_B))$ or $f_A(\pt_{r - 1}(f_A, f_B))$. The total communication cost for this simple protocol is $O(k\cdot\log n)$. However, if Alice and Bob can only communicate $(k - 1)$ rounds, the upper bound then becomes non-trivial.
Nisan and Wigderson \cite{nisan1991rounds} proposed a randomized $(k - 1)$-round protocol with $O((n/k + k)\log n)$ communication bits.

\begin{itemize}
    \item In the beginning, Alice and Bob use public randomness to pick a set of coordinates $I\subseteq[n]$ of size  $10n/k$, and then send $f_A(I)$ and $f_B(I)$ to the other party.
    \item On the other hand, Alice and Bob also simulate ($r$ rounds) deterministic protocol but skip one round if one party finds that the pointer is located in $I$. 
    \item If the skip round never happens, Alice and Bob simply abort at the last round. A simple calculation shows the probability of this event is low.
\end{itemize}

This randomized protocol is indeed very simple. Alice and Bob only share coordinate-wise information. In fact, this is a structured rectangle in our setting.

\paragraph{Lower Bounds.}
Consider $(k - 1)$ round protocols where Alice speaks first.
For deterministic protocols, 
Nisan and Wigderson \cite{nisan1991rounds} proved an $\Omega(n - k\log n)$ communication lower bound. 
In the same paper, they also proved an $\Omega(n / k^2 - k\log n)$ communication lower bound for protocols that achieve $2/3$ accuracy \textit{under uniform input distribution}. 

Since then, lower bounds for pointer chasing and its close variants have been substantially studied by a good amount of papers \cite{duris1984lower,damm1998some,ponzio2001communication,klauck2000quantum,klauck2007interaction,feigenbaum2009graph,guruswami2016superlinear,assadi2019polynomial}. Finally, Yehudayoff \cite{yehudayoff2020pointer} proved an $\Omega(n / k - k\log n)$ lower bound for protocols achieving constant advantage under uniform input distribution. 

Now the main gap between the upper bound \cite{nisan1991rounds} and the lower bound \cite{yehudayoff2020pointer} is the extra $k\log n$ term. This gap becomes significant if $k\geq \sqrt{n}$. In this paper, we further improve the lower bound and close the gap.

\subsection{Our results}
We prove that any protocol that achieves constant advantage under uniform input distribution must communicate $\Omega(n/k + k)$ bits. 
\begin{theorem} \label{thm:main:distributional}
    Let $\Pi$ be a $(k - 1)$-round deterministic protocol for $\chase_k$ where Alice speaks first such that 
    $$
        \pr[f_A, f_B \gets {[n]^n}]{\Pi(f_A, f_B) = \chase_k(f_A, f_B)} \geq 2/3.
    $$
    Then the communication complexity of $\Pi$ is $\Omega(n / k + k)$.
\end{theorem}

By Yao's minimax principle, it implies a lower bound for the $(k - 1)$ round randomized communication complexity.

\begin{corollary}
\label{thm:main}
Every $(k - 1)$-round randomized protocol for $\chase_k$ with error at most $1/3$ (where Alice speaks first) has communication complexity $\Omega(n/k + k)$.
\end{corollary}

We observe there is still a $(\log n)$ gap between our lower bound and the protocol by \cite{nisan1991rounds}. We conjecture that our lower bound is tight and there is a chance to remove the $\log n$ factor in the upper bound side. A simple deterministic protocol with $(k - 1)$ rounds and $O\left(n \right)$ communication bits could be the following: Alice and Bob send the parity of $f_A(x)$ and $f_B(x)$ for all $x\in[n]$ in the beginning. Hence they can skip the last round as they already know the parity. This simple protocol shows that \cite{nisan1991rounds}'s protocol is not tight when $k=o(\log n).$ We believe similar ideas could be extended for large $k$.

\paragraph{Applications.} 
Given the connections between $\chase_k$ and diverse applications \cite{feigenbaum2009graph,nanongkai2011tight,canonne2018adaptivity,joseph2020exponential,chen2022memory,peng2024limitations}, our improved lower bounds automatically lead to several applications. We list two applications below.

\begin{corollary}[Direct sum extension of pointer chasing]
The $(k-1)$-round randomized communication complexity of $\chase_k$ with $d$ pairs of functions is $\Omega(d\cdot n/k^2 + d)$
\end{corollary}
\noindent This corollary improves the previous $\Omega(d\cdot n/k^3 - dk\log n -2d)$ lower bound presented in \cite{feigenbaum2009graph}, which has applications in BFS trees streaming lower bound.

\begin{corollary}[Exponential separations in local differential privacy]
Let $A$ be a $(k-1)$-round sequentially interactive $\varepsilon$-locally private protocol solving $\chase_k$ with error probability $\gamma \leq 1/3$ . Then the sample complexity of $A$ is $\Omega\left(\frac{1}{e^{\varepsilon}} \cdot (n/k + k)\right)$ and there is a $k$ round protocol with sample complexity $\widetilde{O}\left(\frac{k\log n}{\varepsilon^2}\right)$.
\end{corollary}
\noindent This corollary improves the previous $\Omega\left(\frac{n}{e^{\varepsilon} k^2}\right)$ lower bound for $k<\sqrt{n / \log n}$ given 
by \cite{joseph2020exponential}.

\subsection{Gadgetless Lifting: A New Framework to Prove Communication Lower Bounds}
The following two-step approach for proving communication lower bounds often appears in previous works (e.g., \cite{goldmann1992simple, raz1997separation}):
\begin{enumerate}
    \item Identify a family of structured protocols.
    \item Simulate general protocols by structured protocols and prove communication lower bounds for structured protocols.
\end{enumerate}
This approach culminates in query-to-communication lifting theorems \cite{goos2015deterministic, BPPLifting,chattopadhyay2019query,lovett2022lifting}.

\paragraph{Query-to-communication lifting theorems.}
Let $f: Z^n \to \zo$ be a function, and let $g: X \times Y \to Z$ be a two-party gadget function. 
The goal is to prove communication lower bounds for the function $f \circ g^n: X^n \times Y^n \to \zo$. 
Indeed, all functions for which lower bounds are proven using the above approach can be written as $f \circ g^n$ for appropriate $f$ and $g$.
For such functions, a communication protocol can always simulate a decision tree that computes $f$ --- such protocols consist of a natural family of structured protocols. 
Communication complexity for such protocols is essentially the query 
complexity of $f$, for which lower bounds are often easy to prove. 
Hence, the primary job is to show how to simulate general protocols by structured ones. 

Though query-to-communication lifting is a beautiful framework, it requires a gadget function $g$ since $f$ is a one-party function. 
As a consequence, this framework only applies to \textit{lifted functions}, namely, functions that can be written as $f \circ g^n$.
Many important problems, such as pointer chasing, do not fall into this category; hence, lifting theorems do not apply in those cases. 

To address this limitation, we propose a new framework called \textit{gadgetless lifting}. We take a step back to the original approach, reconsidering the choice of structured protocols. 
In some cases, although the function is not a lifted function, 
there are simple and natural protocols.
The crux of gadgetless lifting is how to decide the structured protocols. 
In this paper, we capture it as those protocols that ``all shared useful information are local information''. 
For example, the protocol by \cite{nisan1991rounds} only share local information such as $f_A(x)$ or $f_B(x)$ for some $x\in[n]$. 
In lemma \ref{lem:invariant}, we show that any protocol for $\textsc{PC}_k$ can be simulated by such protocols.
Our proof is inspired by the structure-vs-pseudorandomness decomposition by G{\"o}{\"o}s, Pitassi, and Watson \cite{BPPLifting} and Yang and Zhang \cite{YZ24}, which is a powerful tool that emerged in the study of query-to-communication lifting theorems.
Therefore, we call our method `gadgetless lifting'.

In the study of lifted functions, it has been shown that query-to-communication lifting theorems bypassed some fundamental barriers from previous methods. 
Similarly, gadgetless lifting can also bypass obstacles from existing methods. We discuss two of them below.

\paragraph{Avoiding the loss in round elimination method.}
Previously, the only method to prove round-communication trade-offs is the \textit{round elimination method} \cite{nisan1991rounds}. 
In \cite{nisan1991rounds} and \cite{yehudayoff2020pointer}, the authors studied the pointer chasing problem via the round elimination method.
Denote by $\rv{M}_1, \ldots, \rv{M}_t$ the messages sent in the first $t$ rounds, and let $\rv{Z}_i$ be the pointer in the $i$-th round, \ie, $Z_i = \msf{pt}_{i}(X , Y)$ where $X , Y$ are uniformly chosen from $[n]^n$. 
As is standard the round elimination method, \cite{nisan1991rounds,yehudayoff2020pointer} analyzed the random variables
$$
\rv{R}_t= (\rv{M}_1,\ldots,\rv{M}_t,\rv{Z}_1,\ldots,\rv{Z}_{t-1}) \text{  for $t \leq k$}.
$$
They proved that $\ent(\rv{R}_k) \geq \Omega(n/k)$. Together with the fact that $\ent(\rv{Z}_1, \dots, \rv{Z}_k)=k\log n$, it implies that $\ent(\rv{M}) \geq  \Omega(n/k - k\log n)$. 
The $(k\log n)$ loss (or something similar) appears in many previous works that adopt round elimination-based \cite{nisan1991rounds,klauck2000quantum,klauck2007interaction,guha2007lower,feigenbaum2009graph,yehudayoff2020pointer}. 
In this paper, we avoid the $k\log n$ loss via the gadgetless lifting.

\paragraph{Breaking square-root loss barrier in information complexity.}
Another popular method in proving communication lower bounds is by way of information complexity. However, as mentioned by Yahudayoff \cite{yehudayoff2020pointer}, entropy-based analyses are likely to induce a square-root loss barrier. This barrier usually comes from applying Pinsker’s inequality (or its variant) to bound statistical distance from a small entropy gap. As a consequence, many results such as \cite{nisan1991rounds} can only prove an $\Omega(n/k^2-k\log n)$ lower bound.

As mentioned in \cite{yehudayoff2020pointer}, the square-root loss also appears in many works when using the entropy-based method to prove lower bounds.
For example, it appears in the parallel repetition theorem and is related to the `strong parallel repetition' conjecture which is motivated by Khot’s unique games conjecture \cite{khot2002power}. 
This loss also appears in direct-sum theorems \cite{barak2010compress} and direct-product theorems \cite{braverman2013direct} in communication complexity.

\cite{yehudayoff2020pointer} overcomes this square-root loss barrier by using a non-standard measurement called triangular discrimination. 
By contrast, our approach overcomes the barrier more naturally without using entropy. 

\paragraph{Potentail applications.} 
We noticed that our method can also be naturally extended to multiparty settings such as the numbers in hand model. 
Moreover, some important open problems, such as round-communication tradeoff of bipartite matching problem \cite{blikstad2022nearly} and set pointer chasing problem \cite{feigenbaum2009graph,guruswami2016superlinear}, are difficult to solve using the round elimination method due to its inherent limitations.  Our method offers the potential to solve these challenging problems. 
\section{Preliminaries}

\paragraph{Notations.}
We use capital letters $X$ to denote a set and use bold symbols like $\rv{R}$ to denote random variables. Particularly, for a set $X$, we use $\rv{X}$ to denote the random variable uniformly distributed over the set $X$.  We use $\gets$ to denote sampling from a distribution or choosing an element from a set uniformly at random.

\subsection{Density-Restoring Partition}

\paragraph{Min-entropy and dense distribution.}

For a random variable $\rv{X}$, we use $\supp(\rv{X})$ to denote the support of $\rv{X}$. 

\begin{definition}[Min-entropy and deficiency]
The min-entropy of a random variable $\rv{X}$ is defined by
\[
\ent_\infty(\rv{X}):=\min_{x \in \supp(\rv{X})} \log\left(\frac{1}{\Pr[\rv{X}=x]}\right).
\]
Suppose that $\rv{X}$ is supported on $[n]^J$.
We define the \textit{deficiency} of $\rv{X}$ as 
$$
    \defi(\rv{X}):= \size{J} \log n  - \ent_\infty(\rv{X}).
$$
\end{definition}

For $I \subseteq J$, $x \in [n]^{J}$, let $x(I) \eqdef (x(i))_{i \in I} \in [n]^I$ be the projection of $x$ on coordinates in $I$.

\begin{definition}[Dense distribution]
Let $\gamma \in (0, 1)$.
A random variable $\rv{X}$ supported on $[n]^J$ is said to be \textit{$\gamma$-dense} if for all nonempty $I \subseteq J$, $\ent_\infty(x(I))\geq\gamma |I|\log n$.
\end{definition}


The following lemma is the crux of the structure-vs-pseudorandomness method by \cite{BPPLifting}. It essentially says that a flat random variable could be decomposed into a convex combination of flat random variables with disjoint support and dense properties.

\begin{lemma}[Density-restoring partition] 
\label{lemma:density:restoring:partition}
Let $\gamma \in (0, 1)$.
Let $X$ be a subset of $[n]^M$ and $J \subseteq [M]$.
Suppose that there exists an $\beta \in [n]^{\ol{J}}$ such that $\forall x\in X, x(\ol{J})=\beta$. 
Then, there exists a partition 
$X = X^1\cup X^2\cup\cdots\cup X^r$
and every $X^i$ is associated with a set $I_i \subseteq J$ and a value $\alpha_i \in [n]^{I_i}$ that satisfy the following properties.
\begin{enumerate}
    \item $\forall x\in X^i,x(I_i)= \alpha_i$;
    \item $\rv{X}^i(J\setminus I_i)$ is $\gamma$-dense;
    \item $\defi\left(\rv{X}^i(J \setminus I_i)\right)\leq \defi\left(\rv{X}(J)\right)- (1 - \gamma)\log n \cdot |I_i| + \delta_i$, where $\delta_i \eqdef \log(|X|/|\cup_{j\geq i}X^j|)$. 
\end{enumerate} 
\end{lemma}

\noindent The proof of this lemma, simple and elegant, is included in the appendix for completeness.

\subsection{Communication Protocols}
We recall basic definitions and facts about communication protocols.

\paragraph{Protocol Tree.}
Let $X$ and $Y$ be the input space of Alice and Bob respectively. 
A deterministic communication protocol $\Pi$ is specified by a rooted binary tree.
For every internal vertex $v$, 
\begin{itemize}
    \item it has 2 children, denoted by $\Pi(v, 0)$ and $\Pi(v, 1)$;
    \item $v$ is owned by either Alice or Bob --- we denote the owner by $\msf{owner}(v)$;
    \item every leaf node specifies an output.
\end{itemize}
Starting from the root, 
the owner of the current node $\msf{cur}$ partitions its input space into two parts $X_0$ and $X_1$, and sets the current node to $\Pi(\msf{cur}, b)$ if its input belongs to $X_b$.

\begin{fact}
    The set of all inputs that leads to an internal vertex $v$ is a rectangle, denoted by $\Pi_v = X_v \times Y_v \subseteq X \times Y$.
\end{fact}

The \textit{communication complexity} of $\Pi$, denoted by $\CC(\Pi)$, is the depth of the tree. 
The \textit{round complexity} of $\Pi$, is the minimum number $k$ such that 
in every path from the root to some leaf, the owner switches at most $(k - 1)$ times.
Clearly, if a protocol has $k$ round, then its communication complexity is at least $k$.
We can safely make the following assumptions for any protocol $\Pi$:
\begin{itemize}
    \item $\Pi$ has $k$ rounds on every input; and 
    \item $\Pi$ communicates $\CC(\Pi)$ bits on every input.
\end{itemize}
Indeed, for any protocol, we can add empty messages and rounds in the end, which boosts the communication complexity by a factor of 2.
\section{Proof of Main Theorem }\label{sec: proof}
\begin{theorem}[Main theorem, \cref{thm:main:distributional} restated]
    Let $\Pi$ be a $(k - 1)$-round deterministic protocol for $\chase_k$ where Alice speaks first such that 
    $$
        \pr[f_A, f_B \gets {[n]^n}]{\Pi(f_A, f_B) = \chase_k(f_A, f_B)} \geq 2/3.
    $$
    Then the communication complexity of $\Pi$ is $\Omega(n / k + k)$.
\end{theorem}




We use a \textit{decomposition and sampling process} $\msf{DS}$, as shown in Algorithm ~\ref{algo:main}, in our analysis.
$\msf{DS}$ takes as input a protocol $\Pi$, 
and samples a rectangle $R$ that is contained in $\Pi_v$ for some leaf node $v$.
Our proof proceeds in three steps:
\begin{enumerate}
    \item First, \cref{sec:invariant} analyzes crucial invariants during the running of $\msf{DS}$.
    \item Next, \cref{sec:accuracy} shows that the accuracy of $\Pi$ is captured by a quantity called \textit{average fixed size}, which is a natural quantity that arises in the running of $\msf{DS}$.
    \item Finally, \cref{sec:fixed:size} proves that the average fixed size can be bounded from above by $O(\CC(\Pi))$. Consequently, if $\Pi$ enjoys high accuracy, we get a lower bound of $\CC(\Pi)$.
\end{enumerate}

\subsection{The Decomposition and Sampling Process}  \label{sec:invariant}

During the sampling process, we maintain a useful structure of $R$ mainly by a partitioning-then-sampling mechanism:
At the beginning, $R$ is set to be the set of all inputs. 
Walking down the protocol tree, we decompose the rectangle into structured sub-rectangles; then we sample a decomposed rectangle with respect to its size.
In the end, we arrive at a leaf node $v$ and a subrectangle of $\Pi_v$.

\begin{algorithm}[htp] \label{algo:main}
  \KwIn{A protocol $\Pi$ for the problem $\chase_k$.}
  \KwOut{A rectangle $R=X\times Y$, and $J_A,J_B \subseteq [n]$. }
   Initialize $v:= \text{root of }\Pi, r := 1,  X := Y := [n]^n, J_A:= J_B := [n], \msf{bad} := \false.$\\
  \While{$v$ is not a leaf node}{
        \CommentSty{//Invariant: 
    (1) $X \times Y \subseteq \Pi_v$; (2) there exists some $z_{r - 1} \in [n]$ such that $\pt_{r - 1}(f_A, f_B) = z_{r - 1}\ \forall (f_A, f_B) \in X \times Y$ (See \cref{lem:invariant}).} \\
    Let $u_0 := \Pi(v, 0), u_1 := \Pi(v , 1)$ be the two children of $v$. \\

    \If{$\msf{owner}(v) = \msf{Alice}$}{
        Partition $X$ into $X = X^0 \cup X^1$ such that $X^b \times Y \subseteq \Pi_{u_{b}}$ for $b \in \zo$. \\
        Sample $\rv{b} \in \zo$ such that $\pr{\rv{b}  = b} = |X^{b}|/|X|$ for $b \in \zo$. \\
        Update $X := X^{\rv{b}}, v := u_{\rv{b}}$. \\
        \If{$\msf{owner}(u_{\rv{b}}) = \msf{Bob}$}{ 
            \CommentSty{// A new round.} \\
            Further Partition $X$ into $X = X^{0} \cup X^{1}$ where 
            $X^b := \set{f_A \in X: f_A(z_{r - 1}) \bmod 2 = b}$. \label{line:new:round:partition}\\
            Sample $\rv{b}' \in \zo$ such that $\pr{\rv{b}' = b} = |X^{b}|/|X|$ for $b \in \zo$. \label{line:new:round:update}\\
            Update $X := X^{\rv{b}'}, r := r + 1$. \label{line:new:round}
        } 
        
        Let $X = X^1 \cup \cdots \cup X^m$ be the decomposition of $X$ promised by \cref{lemma:density:restoring:partition} with associated sets $I_1, \dots, I_m \subseteq J_A$. \\ \CommentSty{// Invoking \cref{lemma:density:restoring:partition} with $J = J_A, M = n, \gamma = 1 - \frac{0.1}{\log n}$.}\\
        Sample a random element $\boldsymbol{j} \in [m]$ such that $\Pr[\boldsymbol{j} = j] = |X^j|/|X|$ for $j\in [m]$. \\
        Update $X :=  X^{\boldsymbol{j}}, J_A := J_A \setminus  I_{\boldsymbol{j}}$. \label{line:update:3}\\
        \If{$\msf{owner}(u_{\rv{b}}) = \msf{Bob} \land z_{r - 1} \notin J_B$}{
            $\msf{bad} := \true$.
        }
    }
    \If{$\msf{owner}(v) = \msf{Bob}$}{
        Partition $Y$ into $Y = Y^0 \cup Y^1$ such that $X \times Y^b\subseteq \Pi_{u_{b}}$ for $b \in \zo$. \\
        Sample $\rv{b} \in \zo$ such that $\pr{\rv{b}  = b} = |Y^{b}|/|Y|$ for $b \in \zo$. \\
        Update $Y := Y^{\rv{b}}, v := u_{\rv{b}}$. \\
        \If{$\msf{owner}(u_{\rv{b}}) = \msf{Alice}$}{
            Further Partition $Y$ into $Y = Y^{0} \cup Y^{1}$ where 
            $Y^b := \set{f_B \in Y: f_B(z_{r - 1}) \bmod 2 = b}$. \\
            Sample $\rv{b}' \in \zo$ such that $\pr{\rv{b}' = b} = |Y^{b}|/|Y|$ for $b \in \zo$. \\
            Update $Y := Y^{\rv{b}'}, r := r + 1$. 
        }
        
        Let $Y = Y^1 \cup \cdots \cup Y^m$ be the decomposition of $Y$ promised by \cref{lemma:density:restoring:partition} with associated sets $I_1, \dots, I_m \subseteq J_B$. \\
        Sample a random element $\boldsymbol{j} \in [m]$ such that $\Pr[\boldsymbol{j} = j] = |Y^j|/|Y|$ for $j\in [m]$. \\
        Update $Y :=  Y^{\boldsymbol{j}}, J_B := J_B \setminus  I_{\boldsymbol{j}}$. \\
        
        \If{$\msf{owner}(u_{\rv{b}}) = \msf{Alice} \land z_{r - 1} \notin J_A$}{
            $\msf{bad} := \true$.
        }
    }
  }
  \caption{The decomposition and sampling process $\msf{DS}$}
\end{algorithm}

\begin{lemma}[Loop invariant] \label{lem:invariant}
    Set $\gamma \eqdef 1 - \frac{0.1}{\log n}$. Then in the running of $\msf{DS}(\Pi)$, we have the following loop invariants: 
    After each iteration, 
    \begin{itemize}
        \item[($\diamondsuit$)] $X \times Y \subseteq \Pi_v$; 
        \item[($\clubsuit$)] $\rv{X}(J_A), \rv{Y}(J_B)$ are $\gamma$-dense;
        \item[($\heartsuit$)] there exists some $\alpha_A \in [n]^{\ol{J_A}}, \alpha_B \in [n]^{\ol{J_B}}$ such that $x(\ol{J_A}) = \alpha_A, y(\ol{J_B}) = \alpha_B$ for all $x \in X, y \in Y$;
        \item[($\spadesuit$)] there exists some $z_{r} \in [n]$ such that $\pt_{r}(f_A, f_B) = z_r$ for all $(f_A, f_B) \in X \times Y$.
    \end{itemize}
\end{lemma}
\begin{proof}
    Item ($\diamondsuit$) is true because every time $v$ is updated, 
    $X \times Y$ is updated accordingly to a sub-rectangle of $\Pi_v$ and updating $X \times Y$ into its sub-rectangles does not violate this condition.

    Since we applied density restoring partition at the end of each iteration, Item ($\clubsuit$) and $(\heartsuit)$ is guaranteed by \cref{lemma:density:restoring:partition} and the way that $X, Y, J_A, J_B$ are updated.

    We prove the last item ($\spadesuit$) by induction. 
    Assume that the statement holds after the first $(t - 1)$ iterations. 
    WLOG, assume that at the beginning of the $t$-th iteration, $v$ is owned by Alice. 
    Consider the following two cases.
    \begin{itemize}
        \item \underline{Case 1. Not a new round:} Line \ref{line:new:round} is not executed in the $t$-th iteration. Since $r$ remains unchanged and we only update $R$ to be a sub-rectangle of itself, the statement still holds.
        \item \underline{Case 2. A new round begins: } Line \ref{line:new:round} is executed and $r$ is increased by 1. Let $\rho$ denote the value of $r$ before Line \ref{line:new:round}, then after this iteration, we have $r = \rho  + 1$.
        The induction hypothesis guarantees that 
        there exists some $z_{\rho - 1} \in [n]$ such that
        $$
           \msf{pt}_{\rho - 1}(f_A, f_B) = z_{\rho - 1} 
            \text{ for all} (f_A, f_B) \in X \times Y.
        $$
        Due to the partition and the update in Line \ref{line:new:round:partition} and Line \ref{line:new:round:update}, 
        $|\supp (\rv{X}(z_{\rho - 1}))| \leq n / 2$.
        Hence, $\rv{X}(z_{\rho - 1})$ cannot be $\gamma$-dense as we set $\gamma = 1 - \frac{0.1}{\log n}$.
        Observe that after the update in Line \ref{line:update:3}, 
        $\rv{X}(J_A)$ is $\gamma$-dense. 
        Consequently, we must have $z_{\rho - 1} \in \ol{J_A}$, 
        and by item $(\heartsuit)$, there exists some $z_{\rho} \in [n]$ such that $f_A(z_{\rho - 1}) = z_\rho \ \forall f_A \in X$. 
        By definition, for all $(f_A, f_B) \in X \times Y$, 
        $$
           \msf{pt}_{\rho}(f_A, f_B) = f_A(\msf{pt}_{\rho - 1}(f_A, f_B)) 
            = f_A(z_{\rho - 1}) = z_\rho.
        $$
        This is exactly the same statement after the $t$-th iteration (as we have $r = \rho + 1$). 
    \end{itemize}
\end{proof}

The restricted rectangles in this loop invariant are inspired by the protocols of Nisan and Wigderson \cite{nisan1991rounds}. This lemma aims to capture the fact that Alice and Bob cannot get any additional useful information other than coordinate-wise information during their communication.

\subsection{Relating Accuracy and Average Fixed Size} \label{sec:accuracy}

From \cref{lem:invariant} we know that the coordinates in $\ol{J_A}$ and $\ol{J_B}$ are fixed if we only look at the inputs in $X \times Y$.
Intuitively, the advantage of the protocol comes from such fixed coordinates, since the `alive' coordinates $J_A, J_B$ are dense in the sense that $\rv{X}(J_A), \rv{Y}(J_B)$ is $\gamma$-dense. This intuition is formalized in the following lemma.

\begin{lemma}[Relating accuracy and avarage fixed size] \label{lem:accuracy}
Let $\Pi$ be a $(k - 1)$-round deterministic protocol where Alice speaks first. 
Then 
$$
    \pr[f_A, f_B \gets {[n]}^n]{\Pi(f_A, f_B) = \chase_k(f_A, f_B)} \leq \frac{n^{1 - \gamma}}{2} +  n^{-\gamma} \cdot (k - 1) \cdot \ex[(R, J_A, J_B) \gets \msf{DS}(\Pi)]{|\ol{J_A}| + |\ol{J_B}|}.
$$
\end{lemma}

The proof of the lemma is by the following two claims. 
The first claim readily says that conditioned on the flag $\msf{bad}$ is not
raised, $\Pi$ has little advantage in the rectangle $R$ output by $\msf{DS}(\Pi)$.
The second claim shows the probability that the flag is raised is bounded in terms of the average fixed size.

\begin{claim} \label{claim:hard:part}
    If $\msf{DS}(\Pi)$ outputs $(R = X \times Y, J_A, J_B)$ and $\msf{bad} = \false$ in the end, then 
    $$
        \pr[(f_A, f_B) \gets R]{\Pi(f_A, f_B) = \chase_k(f_A, f_B)} \leq \frac{n^{1 - \gamma}}{2}.
    $$
\end{claim}
\begin{claim} \label{claim:bad:probability}
    $\pr[\msf{DS}(\Pi)]{\msf{bad} = \true} \leq n^{-\gamma} \cdot (k - 1)\cdot \ex[(R, J_A, J_B) \gets \msf{DS}(\Pi)]{|\ol{J_A}| + |\ol{J_B}|}$.
\end{claim}

Next, we first prove \cref{lem:accuracy} using the above two claims, and the proof of the claims is followed. 
\begin{proof}[Proof of \cref{lem:accuracy}]
    Note that in the running of $\msf{DS}(\Pi)$, 
    we always update $R$ to a randomly chosen rectangle and the probability of each rectangle being chosen is proportional to its size. 
    Consequently, 
    \begin{align*}
        &\quad \pr[f_A, f_B \gets {[n]^n}]{\Pi(f_A, f_B) = \chase_k(f_A, f_B)} \\
        &= \pr[(R, J_A, J_B) \gets \msf{DS}(\Pi), (f_A, f_B) \gets R]{\Pi(f_A, f_B) = \chase_k(f_A, f_B)} \\
        &\leq \pr[\msf{DS}(\Pi)]{\msf{bad} = \true} + \pr[(R, J_A, J_B) \gets \msf{DS}(\Pi), (f_A, f_B) \gets R]{\Pi(f_A, f_B) = \chase_k(f_A, f_B) \land \msf{bad} = \false} \\
        &\leq \frac{n^{1 - \gamma}}{2} +  n^{-\gamma} \cdot (k - 1) \cdot \ex[(R, J_A, J_B) \gets \msf{DS}(\Pi)]{|\ol{J_A}| + |\ol{J_B}|}.
    \end{align*}
    where the last step is by \cref{claim:hard:part} and \cref{claim:bad:probability}.
\end{proof}

It remains to prove the two claims.

\begin{proof}[Proof of \cref{claim:hard:part}]
    WLOG, assume $k - 1$ is odd and the protocol always has $k$ round.
    Let $z_{k - 1}$ be the pointer guaranteed by the loop invariant (\cref{lem:invariant}), \ie $\pt_{k - 1}(f_A, f_B) = z_{k - 1}$ for all $(f_A, f_B) \in R$. Since $\msf{bad} = \false$, we have $z_{k - 1} \in J_A$. Again by the loop invariant, $\ent_\infty(\rv{X}(z_{k - 1})) \geq \gamma$. 
    Moreover, since $R$ is contained in some leaf node of $\Pi$, 
    $\Pi$ output the same answer in $R$, say $b^* \in \zo$.
    Consequently, 
    \begin{align*}
        \pr[(f_A, f_B) \gets R]{\Pi(f_A, f_B) = \chase_k(f_A, f_B)} 
        &= \pr[f_A \gets X]{f_A(z_{k - 1}) \bmod 2 = b^*} \\
        &\leq \sum_{\sigma \in [n] : \sigma \bmod 2 = b^*} \pr[f_A \gets X]{f_A(z_{k - 1}) = \sigma} \\
        &\leq \frac{n}{2} \cdot n^{-\gamma}.
    \end{align*}
\end{proof}

\begin{proof}[Proof of \cref{claim:bad:probability}]
    Let $\ee_\ell$ denote the event that the flag $\msf{bad}$ is raised when $r = \ell + 1$ (\ie when the $\ell$-th round ends) for the first time.
    Clearly, 
    $
        \pr{\msf{bad} = \true} = \sum_{\ell = 1}^{k - 1} \pr{\ee_\ell}.
    $
    It suffices to bound each $\pr{\ee_\ell}$. 

    Assume $\ell$ is odd, meaning that Alice speaks in the $\ell$-th round.
    Let $\msf{coin}$ denote the randomness used for the first $(\ell - 1)$
    rounds. 
    Let $X^{(\ell - 1)}$, $J_A^{(\ell - 1)}, J_B^{(\ell - 1)}$ be the sets $X, J_A, J_B$ when executing $\msf{DS}(\Pi)$ using $\msf{coin}$ until the $\ell$-th round begins.
    Let $z_{\ell - 1}$ be the pointer promised by the invariant.
    For $\ee_\ell$ to happen, we must have $\msf{bad} = \false$ until the $\ell$-th round begins, meaning that $z_{\ell - 1} \in J_A^{(\ell - 1)}$.

    Note that the random variable $\rv{z}_\ell$ exactly has the same distribution as $\rv{X}^{(\ell - 1)}(z_{\ell - 1})$.
    This is because, in the $\ell$-th round (\ie until $r$ steps to $\ell + 1$), we decompose $X^{(\ell - 1)}$ into finer sets and update $X$ to be one of them with probability proportional to their size. Therefore, 
    \begin{align*}
        \pr[\msf{coin}']{\ee_t} = \pr[\msf{coin}']{\rv{z}_\ell \notin J_B^{(\ell - 1)}}
        = \pr[f_A \gets \rv{X}^{(\ell - 1)}]{f_A(\rv{z}_{\ell - 1} ) \notin J_B^{(\ell -1)}}
        &= \sum_{\sigma \in \ol{J_B^{(\ell - 1)}}} \pr[f_A \gets \rv{X}^{(\ell - 1)}]{f_A(\rv{z}_{\ell - 1})= \sigma} \\
        &\leq \size{\ol{J_B^{(\ell - 1)}}} \cdot n^{-\gamma},
    \end{align*}
    where we fix $\msf{coin}$ and the probability runs over $\msf{coin}'$, the randomness used afterward; the last inequality holds because $z_{\ell - 1} \in J_A^{(\ell - 1)}$ and $\rv{X}^{(\ell - 1)}\left(J_A^{(\ell - 1)}\right)$ is $\gamma$-dense (by Item ($\clubsuit$) in \cref{lem:invariant}).
    Averaging over $\msf{coin}$, we get 
    $$
         \pr[\msf{DS}(\Pi)]{\ee_\ell} \leq \ex[\msf{coin}]{\size{\ol{J_{B}^{(\ell - 1)}}}} \cdot n^{-\gamma} \leq \ex[(R, J_A, J_B) \gets \msf{DS}(\Pi)]{|\ol{J_{B}}|} \cdot n^{-\gamma},
    $$
    where the second inequality holds because $J_B$ becomes smaller and smaller during the execution.
    
    For even $\ell$'s, we analogously have $\pr{\ee_\ell} \leq \ex{|\ol{J_A}|} \cdot n^{-\gamma}$, and hence the claim follows from union bound. 
\end{proof}

\subsection{Average Fixed Size is Bounded by Communication}
\label{sec:fixed:size}

Now that the accuracy of a protocol $\Pi$ is bounded from above by the average fixed size (\ie $\ex[(R, J_A, J_B) \gets \msf{DS}(\Pi)]{|\ol{J_A}| + |\ol{J_B}|}$), in what follows we show that the average fixed size is at most $O(\CC(\Pi))$.
Formally, we prove that 

\begin{lemma} \label{lem:fixed:size}
Let $\Pi$ be a $(k - 1)$-round deterministic protocol where Alice speaks first. 
Then 
$$
    \ex[(R, J_A, J_B) \gets \msf{DS}(\Pi)]{|\ol{J_A}| +  |\ol{J_B}|} \leq \frac{3\CC(\Pi)}{(1 - \gamma)\log n}.
$$
\end{lemma}
\begin{remark}
    We shall set $\gamma := \gammavalue$ and hence the right-handed side equals $30\CC(\Pi)$.
\end{remark}
\begin{proof}
We shall prove this lemma by density increment argument. 
That is, we study the change of the density function 
\begin{equation}
    \label{equ:density:function}
    \defi(R) \eqdef \defi(\rv{X}(J_A)) + \defi(\rv{Y}(J_B)).
\end{equation}
in each iteration.
Let $\rv{\phi_t}$ be the value of $\defi(R)$ at the end of the $t$-th iteration. 
Assume without loss of generality Alice speaks (\ie $\msf{owner}(v) = \msf{Alice}$) in the $t$-th iteration.

We fix the random coins used for the first $(t - 1)$ iterations and consider the updates in the current iteration. 
\begin{enumerate}
    \item First, $X$ is partitioned into $X = X^0 \cup X^1$ according to $\Pi$.
    Then, $X$ is updated to $X^{b}$ with probability $\frac{|X^{b}|}{|X|}$.
    Consequently, $\defi(\rv{X}(J_A))$ will increase as $|X|$ shrinks, and in expectation (over the random choice of $\rv{b}$) the increment is 
    \begin{equation}
    \sum_{b \in \zo} \frac{|X^{b}|}{|X|}\log\bigg(\frac{|X|}{|X^{b}|}\bigg) \leq 1.
    \end{equation}
    \item Next, suppose that updating $v$ leads to the switch of the owner, \ie Line \ref{line:new:round} is triggered. Since we also partition $X$ into two parts and update $X$ with probability proportional to the size of each part, the same argument applies. 
    That is, taking expectation over the random choice of $\rv{b}'$, 
    $\defi(\rv{X}(J_A))$ increases by at most 1 in expectation.
    \item Finally, we further partition $X$ according to \cref{lemma:density:restoring:partition}.
    Say $X$ is partitioned into $X = X^{1} \cup \dots \cup X^{m}$ and let $I_1, \dots, I_m$ be the index sets promised by \cref{lemma:density:restoring:partition}; and for all $j \in [m]$ we have 
    $$
        \defi(\rv{X}^j(J_A \setminus I_j)) \leq \defi(\rv{X}(J_A)) - (1 - \gamma)  \log n |I_j| + \delta_j,
    $$
    where $\delta_j = \log(|X| / \cup_{v \geq j} X^v)$. 
    With probability $p_j \eqdef |X^j| / |X|$, we update $X := X^j$ and $J_A := J_A \setminus I_j$. 
    Therefore, taking expectation over the random choice of $\rv{j}$, the density function will decrease by 
    \begin{equation} \label{equ:decrement:1}
        \defi(\rv{X}(J_A)) - \ex[j \gets \bm{j}]{\defi(\rv{X}^j(J_A \setminus I_j))}
        \geq \ex[j \gets \bm{j}]{(1 - \gamma) \log n \cdot  |I_j| - \delta_j}.
    \end{equation}
    Note that $\delta_j \eqdef \log \frac{1}{\sum_{v \geq j}p_v}$ and thus 
    \begin{equation} \label{equ:decrement:2}
        \ex[j \gets \bm{j}]{\delta_j} = \sum_{j = 1}^m p_j  \log \frac{1}{\sum_{v \geq j}p_j}
        \leq \int_{0}^1 \log \frac{1}{1 - x} \mathrm{d} x \leq 1.
    \end{equation}
\end{enumerate}

Let $\fFam_{t - 1}$ be the $\sigma$-algebra generated by the random coins used for the first $(t - 1)$ iterations.
Let $\rv{\beta}_t$ be the increment of $|\ol{J_A}|$ and $|\ol{J_B}|$ in the $t$-th iteration. 
Observe that $\rv{\beta}_t = |I_{\rv{j}}|$ by definition. 
By \cref{equ:decrement:1} and \cref{equ:decrement:2}, 
taking expectation over random choice of $\rv{j}$, 
$\defi(\rv{X}(J_A))$ decrease by at least $ (1 - \gamma)\log n \cdot  \ex{\rv{\beta}_{t} \mid \fFam_{t - 1}}  - 1$ due to the density restoring partition.
Then 
\begin{equation}
    \label{equ:increment}
    \ex{\rv{\phi}_t - \rv{\phi}_{t - 1}} 
    = \ex{ \ex{\rv{\phi}_t - \rv{\phi}_{t - 1} \mid \fFam_{t - 1}}} 
    \leq \ex{1 + \rv{\eta_t} - ((1 - \gamma)\log n \cdot \rv{\beta}_t- 1)},
\end{equation}
where $\rv{\eta_t} \eqdef \one[\text{owner switches in the $t$-th iteration}]$.

Write $c \eqdef \CC(\Pi)$ and assume we always have $c$ iterations. \footnote{Namely, $\Pi$ communicates $c$ bits on all inputs.}
In the beginning, $\rv{\phi}_0 = \defi([n]^n \times [n]^n) = 0$.
Since the density function is always non-negative by definition, we have $\rv{\phi}_{c} \geq 0$ and thus $\ex{\rv{\phi}_{c} - \rv{\phi}_{0}} \geq 0.$
On the other hand, by telescoping, 
$$
    \ex{\rv{\phi}_{c} - \rv{\phi}_{0}} = \sum_{t = 1}^c \ex{\rv{\phi}_{t} - \rv{\phi}_{t - 1}}
    \leq 2c + \sum_{t = 1}^c \ex{\rv{\eta}_t - (1 - \gamma)\log n \cdot\rv{\beta_t}},
$$
where the inequality follows from \cref{equ:increment}.
Observe that 
$\sum_{t = 1}^c\rv{\eta}_t$ is at most $k$ and 
$\sum_{t = 1}^c\rv{\beta}_t = |\ol{\rv{J}_A}| + |\ol{\rv{J}_B}|$ by definition. 
We conclude that 
$$
     \ex{|\ol{\rv{J}_A}| + |\ol{\rv{J}_B}|} = \ex{\sum_{t = 1}^c \rv{\beta}_t} \leq  \frac{2c + k}{(1 - \gamma)\log n} \leq \frac{3c}{(1 - \gamma)\log n},
$$
as desired.
\end{proof}

\paragraph{Proving the main theorem.} Now our main theorem easily follows from the two lemmas.
\begin{proof}[Proof of \cref{thm:main:distributional}]
    Set $\gamma \eqdef 1 - \frac{0.1}{\log n}$.
    By \cref{lem:fixed:size} and \cref{lem:accuracy}, we get 
    \begin{align*}
        \mathtt{Accuracy}(\Pi) \eqdef \pr[f_A, f_B \gets {[n]^n}]{\Pi(f_A, f_B) = \chase_k(f_A, f_B)}
        &\leq \frac{n^{1 - \gamma}}{2} +  n^{-\gamma} \cdot (k - 1) \cdot  \frac{3\CC(\Pi)}{(1 - \gamma)\log n} \\
        &\leq 0.54 + \frac{1.08 (k - 1)}{n} \cdot 30\CC(\Pi),
    \end{align*}
    where we use 
    $
        \frac{n^{1 - \gamma}}{2} \leq 0.54, n^{-\gamma} \leq \frac{1.08}{n}.
    $
    Since we assumed $\mathtt{Accuracy}(\Pi) \geq 2/3$, we conclude that 
    $$
        \CC(\Pi) \geq \frac{2/3 - 0.54}{1.08 \cdot 30} \cdot \frac{n}{k - 1} > 0.0039 \cdot \frac{n}{k - 1} = \Omega(n / k).
    $$
    We also trivially have $\CC(\Pi) \geq k - 1$ as $\Pi$ has $(k - 1)$ rounds; putting it together we conclude that $\CC(\Pi) = \Omega(n / k + k)$.
\end{proof}

\section*{Acknowledgements}
We thank Sepehr Assadi and Yuval Filmus for their helpful comments.

\bibliographystyle{alpha}
\bibliography{reference.bib}
\newpage
\section*{Appendix}
The following lemma and proof are from Lemma 5 in \cite{BPPLifting}.
\begin{lemma}[\cref{lemma:density:restoring:partition} restated]
Let $\gamma \in (0, 1)$.
Let $X$ be a subset of $[n]^M$ and $J \subseteq [M]$.
Suppose that there exists an $\beta \in [n]^{\ol{J}}$ such that $\forall x\in X, x(\ol{J})=\beta$. 
Then, there exists a partition 
$X = X^1\cup X^2\cup\cdots\cup X^r$
and every $X^i$ is associated with a set $I_i \subseteq J$ and a value $\alpha_i \in \zo^{I_i}$ that satisfy the following properties.
\begin{enumerate}
    \item $\forall x\in X^i,x(I_i)= \alpha_i$;
    \item $\rv{X}^i(J\setminus I_i)$ is $\gamma$-dense;
    \item $\defi\left(\rv{X}^i(J \setminus I_i)\right)\leq \defi\left(\rv{X}(J)\right)- (1 - \gamma)\log n \cdot |I_i| + \delta_i$, where $\delta_i \eqdef \log(|X|/|\cup_{j\geq i}X^j|)$. 
\end{enumerate} 
\end{lemma}

\begin{proof}
We prove it by a greedy algorithm as follows.

\begin{algorithm}[ht]\label{greedy}
  \KwIn{$X\subseteq [n]^M$}
  \KwOut{A partition $X = X^1\cup X^2\cup\cdots\cup X^m$}
  Initialize $i := 1$.
  
  \While{$X \neq \emptyset$}
  {
  Let  $I \subseteq J$ be a maximal subset (possibly $I = \emptyset$) such that $\ent_\infty(\rv{X}(I)) < \gamma|I| \log n$ and let $\alpha_i \in [n]^I$ be a witness of this fact, \ie  $\Pr[\rv{X}(I) = \alpha_i] > n^{-\gamma |I|}$.\\
  $X^i := \{x \in X : x(I) = \alpha_i\}$ and $I_i := I$. \\
  Update $X := X \setminus X^i$, $J := J \setminus I_i$, and $i := i+1$.
  }
  \caption{Greedy Algorithm}
\label{algo:decompose}
\end{algorithm}

Item 1 is guaranteed by the construction of $X^i$ and $I_i$.

We prove Item 2 by contradiction.
Assume towards contradiction that $\rv{X}^i(J\setminus I_i)$ is not $\gamma$-dense for some $i$.
By definition, there is a nonempty set \(K \subseteq J \setminus I_i\) and \(\beta \in [n]^K\) violating the min-entropy condition, namely, 
$
    \pr{\rv{X}(K) = \beta} >  n^{-\gamma |K|}.
$
Write $X^{\geq i} \eqdef \cup_{j \geq i}X^i$. 
Then 
$$
    \pr{\rv{X}^{\geq i}(I_i \cup K) = (\alpha_i, \beta)} 
    =  \pr{\rv{X}^{\geq i}(I_i) = \alpha_i} \cdot  \pr{\rv{X}^{i}(K) = \beta} > n^{-\gamma|I_i|} \cdot n^{- \gamma|K|} = n^{-\gamma|I_i \cup K|},
$$
where the first equality holds as $(\rv{X}^{\geq i} | \rv{X}^{\geq i}(I_i) = \alpha_i) = \rv{X}^i$.
However, this means at moment that $I_i$ is chosen, 
the set \(I_i \cup K \subseteq J\) also violates the min-entropy condition (witnessed by \((\alpha_i, \beta)\)), contradicting the maximality of \(I_i\).

Finally, Item 3 is proved by straightforward calculation:
\begin{equation*}
\begin{aligned}
\defi(\rv{X}^i(J \setminus I_i)) 
&= |J \setminus I_i| \log n - \log |X^i|\\
&\leq (|J| \log n - |I_i| \log n) - \log \left( \size{X^{\geq i}} \cdot n^{-\gamma |I_i|} \right)\\
&= (|J| \log n - \log |X| ) - (1 - \gamma)|I_i| \cdot \log n + \log \left( \frac{|X|}{\size{X^{\geq i}}} \right) \\
&= \defi\Big(\boldsymbol{X}(J)\Big)-(1-\gamma)|I_i|\log n +\delta_i.
\end{aligned}
\end{equation*}

\end{proof}

\end{document}